\documentclass[10pt,twocolumn]{IEEEtran}
\usepackage{amsthm,amssymb,amsmath,bm}
\usepackage{graphicx}
\usepackage{fancyhdr}
\usepackage[font={small}]{caption}
\usepackage{tabularx}
\usepackage{cite}
\usepackage{color,soul}
\usepackage[Algorithm]{algorithm}
\newtheorem{theorem}{Lemma}
\newtheorem{rem}{Remark}
\begin{document}

\title{A New Multiple Access Technique for 5G:  Power Domain Sparse Code Multiple Access (PSMA)}

\author{Nader Mokari, Mohammad.R Javan, Mohammad Moltafet, Hamid Saeedi, and Hossein Pishro-Nik
\thanks{N. Mokari, M. Moltafet, and H. Saeedi are with ECE department, Tarbiat Modares University, Tehrn, Iran. M.R. Javan
 is with School of Electrical and Robotic Engineering, Sharood University of Technology, Shahrood, Iran. H. Pishro-Nik is with ECE department, University of Massachusetts, 
 Amherst, MA, USA.}}

\maketitle

\begin{abstract}
In this paper, a new approach for multiple access (MA) in fifth generation (5G) of cellular networks called  power domain sparse code multiple access (PSMA)  is proposed. In PSMA, we adopt both  the power domain and the code domain to transmit multiple users' signals over a subcarrier simultaneously. In such a model, the same SCMA codebook can be used by multiple users where, for these users, power domain non-orthogonal multiple access (PD-NOMA) technique is used to send  signals non-orthogonally. Although different SCMA codebooks are orthogonal and produce no interference over each other, the same codebook used by multiple users produces interference over these users. We investigate the signal model as well as the receiver and transmitter of the PSMA method. To evaluate the performance of PSMA, we consider a heterogeneous cellular network (HetNet). In this case our design objective is to maximize the system sum rate of the network subject to some system level and QoS constraints such as transmit power constraints. We formulate the proposed resource allocation problem as an optimization problem and solve it by successive convex approximation (SCA) techniques. Moreover, we  compare PSMA with sparse code multiple access (SCMA) and PD-NOMA from the performance and computational complexity perspective. Finally, the effectiveness of the proposed approach is investigated using numerical results. We show that by a reasonable increase in complexity, PSMA can improve the spectral efficiency about 50\% compared to SCMA and PD-NOMA.
\end{abstract}

\section{introduction}
Multiple access (MA) techniques have an essential role on the performance improvement of  cellular networks. In the fourth generation (4G) of cellular networks, orthogonal frequency division  multiple access (OFDMA) is proposed as an efficient MA technique to address  the upcoming challenges. As the statistical data shows, mobile data traffic will grow  several folds in the next decade \cite{int1,int2}. Therefore, the next generations of cellular networks should be designed to address  existing challenges like spectral efficiency (SE) and energy efficiency (EE).

Due to the  demand for high data rate services and the limitations of the available bandwidth for cellular networks, applying new techniques and methods to improve SE in 5G is very important. For MA techniques in 5G, some non-orthogonal techniques such as power domain non-orthogonal multiple access (PD-NOMA) \cite{Saito} and sparse code multiple access (SCMA) \cite{NIC} are proposed. By applying  superimposed coding on the transmitter side, PD-NOMA assigns a subcarrier to multiple users simultaneously, while on the receiver side, by using successive interference cancellation (SIC) method, the signals of users are detected. SCMA is a codebook based MA technique in which each subcarrier can be used in different codebooks on the transmitter side, and on the receiver side, users' signals are detected by applying message passing approach (MPA).

 Recently, PD-NOMA and SCMA have received significant attention as appropriate candidates for MA technique for 5G \cite{pp0}-\cite{int8}.
 In \cite{pp0}, the authors study
 user pairing in a PD-NOMA based system. They show that the system throughput can be improved by pairing
 users enjoying good channel situations with users suffering
 from poor channel conditions.
 The authors of \cite{pp1} study the joint power allocation and precoding
 design in a multiuser multiple input multiple output (MIMO) PD-NOMA based system in order to maximize the system sum rate.
 In \cite{pp8}, by considering a minimum data rate requirement for the users with bad channel situations, the authors propose a power allocation problem. To solve the corresponding problem, they apply  two algorithms:  an optimal algorithm using the bisection search method with high computational complexity and a suboptimal one based on the SIC approach.
 In \cite{pp9}, the authors propose different resource allocation problems. Based on the proposed method, they study the effect of the fairness among  users in a PD-NOMA based system. The authors of \cite{pp10} propose a resource allocation method which maximizes the energy efficiency in an SCMA-based system. In \cite{our}, PD-NOMA and SCMA as the pioneer candidates of MA in 5G are compared from performance and receiver complexity points of view.

In the core of PD-NOMA lies the SIC technique which makes it possible to allocate one subcarrier to more than one user by removing the interference resulting from non-exclusive utilization of subcarriers. On the other hand, SCMA employs exclusive codebooks for each user where each codebook is assigned non-exclusively to certain number of subcarriers. Instead of SIC, SCMA uses MPA to remove the resulting interference. The core idea of PSMA is to take advantage of SCMA ability to remove the inference using MPA while letting the users to utilize the codebooks non-exclusively, i.e., each codebook is assigned to more than one user. The SIC method is then applied to remove the interference resulting from non-exclusive use of the codebooks after applying the MPA.

The main contributions of this paper can be summarized as follows:
\begin{itemize}
\item
We investigate the signal model and detection techniques of PSMA together with PD-NOMA and SCMA for comparison.
\item
We compare the receiver complexity for the three techniques.
\item 
We propose a novel resource allocation problem for PSMA-based heterogeneous networks (HetNets) in which we maximizes the sum rate with certain  constraints. To solve the corresponding problem, we apply advanced techniques to address the non-convexity of the problem. 
\item
For practical scenarios, we consider similar resource allocation problems for PD-NOMA and SCMA to compare their performance and complexity with that of PSMA.
\end{itemize}

The results indicate that by a reasonable increase in the receiver complexity, PSMA can improve the spectral efficiency as far as 50\% compared to SCMA and PD-NOMA which is, in our opinion, a great achievement.

The remainder of this paper is organized as follows:  system model and  signal model for transmitter and receiver of PD-NOMA and SCMA are studied in Section \ref{Multiple Access Technique: PD-NOMA} and Section \ref{Multiple Access Technique: SCMA}, respectively. The PSMA technique and its  signaling on transmitter and receiver sides  are investigated in  Section \ref{Multiple Access Technique: PSMA}.
The complexity of the receiver for different MA techniques  is studied  in Section \ref{Receiver Complexity}. To evaluate the performance of  PSMA,
 resource allocation problems are proposed in Section
\ref{Optimization Problems for PSMA based Cellular Networks}. In Section \ref{solution of the joint power and subcarrier allocation}, the solution algorithms of the proposed resource allocation problems are presented. Numerical results are presented in Section \ref{simulation resuls}. Finally, the paper is concluded  in Section \ref{CONCLUSION}.

\section{Multiple Access Technique: PD-NOMA}\label{Multiple Access Technique: PD-NOMA}
In a  PD-NOMA system, each subcarrier can be assigned to multiple users  simultaneously by applying superimposed coding (SC), and each user removes the signals of other users by exploiting SIC. Based on the PD-NOMA approach, each user   on the receiver side   removes the signals of the users with worse channel, and considers the signals of other users  as noise.   
\begin{figure*}[t]
\centering
\includegraphics[width=1\textwidth]{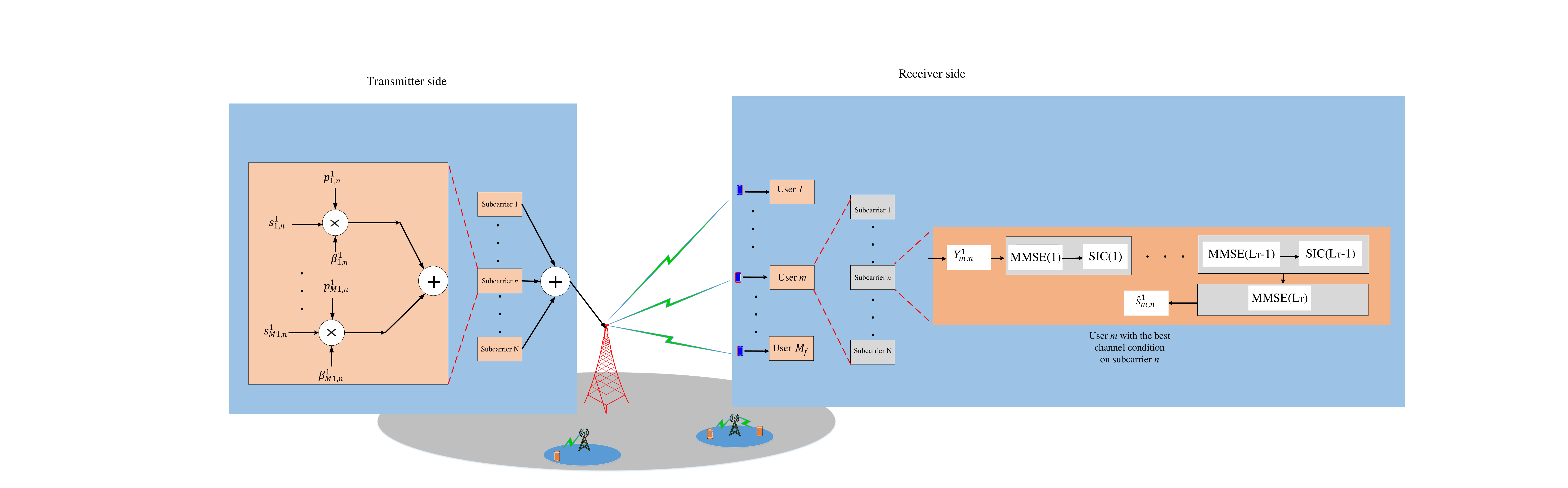}
\caption{Block diagram of transmitter and receiver in a PD-NOMA-based system.}
\label{NOMA}
\end{figure*}
\subsection{PD-NOMA System Model}
We consider a downlink PD-NOMA-based HetNet with $F$ BSs, $M$ users, and $N$ subcarriers. In this system model,
$\mathcal{F}$ indicates the set of BSs with $f=1$ denoting the macro base station  (MBS). $\mathcal{M}_f$ denotes the set of  users of cell $f$ with $\bigcup_{f\in\mathcal{F}}\mathcal{M}_f=\mathcal{M}$ denoting the set of all users. $\mathcal{N}$ demonstrates the set of subcarriers. In addition, $h_{m,n}^f$ is the channel coefficient between user $m$ and BS $f$ on subcarrier $n$ and $p_{m,n}^f$ is the transmit power of BS $f$ to user $m$ over subcarrier $n$. We define the subcarrier assignment indicator  $\beta^f_{m,n}$ with $\beta^f_{m,n}=1$  if subcarrier $n$ is assigned to user $m$ over BS $f$, and otherwise $\beta^f_{m,n}=0$.

\subsection{PD-NOMA Signal Model: Received Signal}
With parameter definitions in the previous subsection, the received signal at user $m$ in  BS $f$ over subcarrier $n$ is given by:
\begin{align}\label{PDsignal}
&Y_{m,n}^f=h^f_{m,n}\sqrt{p^f_{m,n}}s^f_{m,n}+\sum_{j\in\mathcal{M}_f,j\neq m}\beta^f_{j,n} h^f_{m,n}\sqrt{p^f_{j,n}}s^f_{j,n}\\\nonumber&+\sum_{i\in\mathcal{F}, i\ne f}\sum_{j\in\mathcal{M}_i}\beta^i_{j,n} h^i_{m,n}\sqrt{p^i_{j,n}}s^i_{j,n}+w_{m,n}^f,
\end{align}
where $s^i_{j,n}$ indicates the message which BS $i$ sends to user $j$ on subcarrier $n$ and $w_{n,m}^f$ shows the noise. The first term of \eqref{PDsignal} indicates the signal of user $m$ on subcarrier $n$ over BS $f$, the second term  presents the NOMA interference term, and the third term shows the inter-cell interference.

\subsection{PD-NOMA Signal: Receiver Side}
In the PD-NOMA approach, on  the receiver side for user $m$, the signal of users with better channel situation  are considered as noise and the signals of other users are detected and removed. For example, if we consider a worse user $m'$, i.e., $|h^f_{m,n}|^2>|h^f_{m',n}|^2$, the signal to interference plus noise ratio (SINR) of user $m'$ at user $m$, i.e., $\gamma_{m',n}^f(m)$, can be obtained by \eqref{SINRmmprinpdnoma} shown on the next page.
\begin{figure*}[t]
\begin{align}\label{SINRmmprinpdnoma}
&\gamma_{m',n}^f(m)=\frac{|h^f_{m,n}|^2 p^f_{m',n}}{\sum_{j\in \mathcal{M}_f, |h^f_{j,n}|^2\ge|h^f_{m',n}|^2}|h^f_{m,n}|^2 p^f_{j,n}+\sum_{i\in\mathcal{F}, i\ne f}\sum_{j\in\mathcal{M}_f}|h^i_{m,n}|^2 p^i_{j,n}+|w_{m,n}^f|^2}.
\end{align}
\end{figure*}
An important issue in this procedure is  that we assume user $m$ is able to decode the message of user $m'$. This is possible if the SINR of user $m'$ at user $m$ is larger than that of user $m'$ at its receiver, i.e., $\gamma_{m',n}^f(m)\geq \beta_{m,n}^f\gamma_{m',n}^f(m')$. In single cell PD-NOMA, the strategy is to sort users based on their channel gains and let users with better channels decode the signals of users with worse channels. This strategy ensures that the condition $\gamma_{m',n}^f(m)\geq \beta_{m,n}^f\gamma_{m',n}^f(m')$ for single cell PD-NOMA is always satisfied. However, for multicell PD-NOMA. we must adopt appropriate policy to guarantee this condition.

After decoding the signal of user $m'$, user $m$ will decode it and subtract it from the received signal. User $m$ will apply this SIC procedure for all users worse than itself to finally obtain the following signal:
\begin{align}
&\hat{Y}_{m,n}^f=h^f_{m,n}\sqrt{p^f_{m,n}}s^f_{m,n}\\\nonumber&+\sum_{j\in \mathcal{M}_f, |h^f_{j,n}|^2\ge|h^f_{m,n}|^2}\beta^f_{j,n} h^f_{m,n}\sqrt{p^f_{j,n}}s^f_{j,n}\\\nonumber&+\sum_{i\in\mathcal{F}, i\ne f}\sum_{j\in\mathcal{M}_i}\beta^i_{j,n} h^i_{m,n}\sqrt{p^i_{j,n}}s^i_{j,n}+w_{m,n}^f,
\end{align}
which leads to SINR $\gamma_{m,n}^f(m)$ which can be obtained from \eqref{SINRmmprinpdnoma} using the corresponding parameters. The transmitter and receiver block diagram  of a PD-NOMA-based system with minimum mean square error (MMSE) detector  is shown in Fig. \ref{NOMA}, in which $L_T$ shows the number of users which are superimposed in each subcarrier.

\section{Multiple Access Technique: SCMA}\label{Multiple Access Technique: SCMA}
An SCMA encoder is  a mapping from $\log_2 (J)$ bits to an
$N$-dimensional  codebook of size $J$ \cite{NIC}. The
N-dimensional  codewords of a codebook are sparse vectors with $U$
($U < N$) non-zero entries that refer to $U$ specific
subcarriers. Based on the SCMA approach, codebooks that are
composed of subcarriers are the basic resource units in networks
\cite{NIC,int8}. If each
codebook consists of $U$ subcarriers, there are
$C(N,U)=\dfrac{N!}{(N-U)!U!}$ codebooks in the considered system. The transmitter and receiver block diagrams of an SCMA-based system are shown in Fig. \ref{SCMA}.
\begin{figure*}[t]
\centering
\includegraphics[width=.8\textwidth]{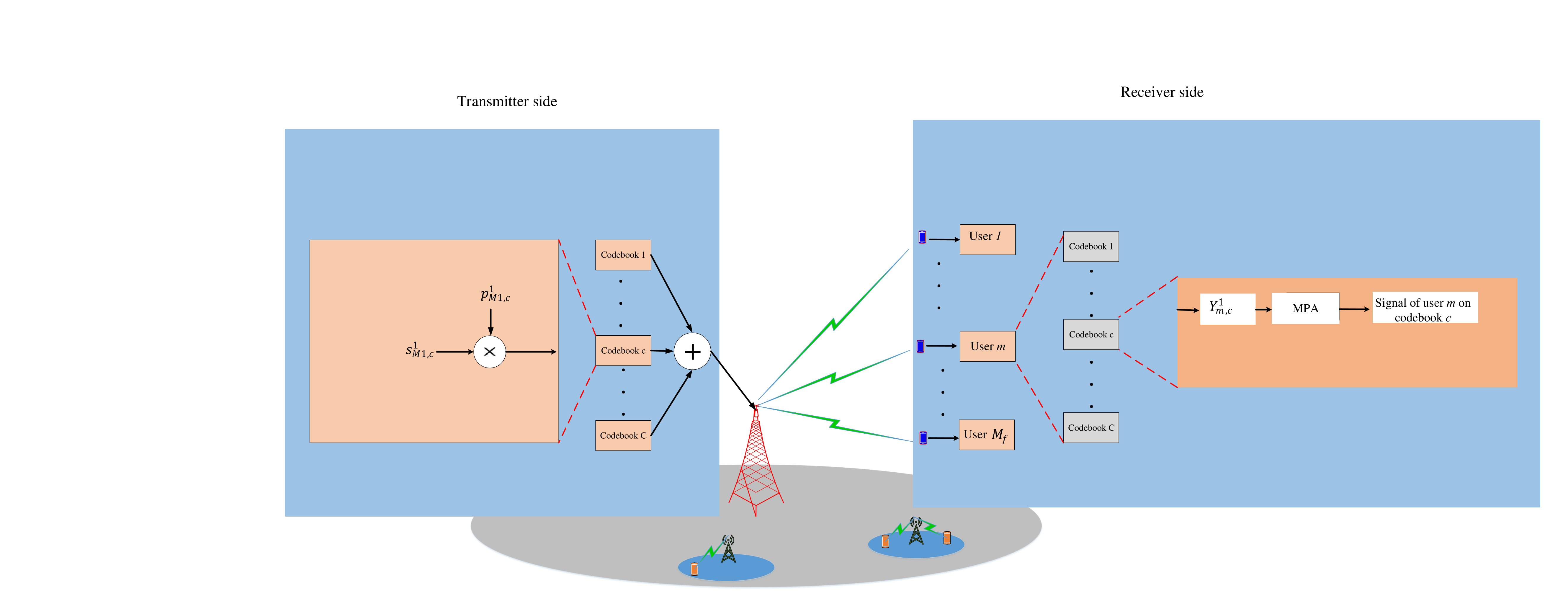}
\caption{Block diagram of transmitter and receiver in a SCMA-based system.}
\label{SCMA}
\end{figure*}

\subsection{SCMA System Model}
In SCMA, the BS sends the signals of all users simultaneously. Since users use different codes that occupy different subcarriers, the received signal over each subcarrier has components from signals of other users' codes. Note that in the considered HetNet model, users in each cell are assumed to use different codes meaning that a code is reused in each cell only once. We consider a downlink SCMA-based HetNet with $F$ BSs, $M$ users, and $C$ codebooks. In this system model, $\mathcal{C}$ indicates the set of codebooks. In addition,  $p_{m,c}^f$ is the transmit power of BS $f$ to user $m$ on codebook $c$, and $q^f_{m,c}$ is the codebook assignment between user
$m$ and codebook $c$  in BS $f$ with $q^f_{m,c}=1$ if codebook $c$
is allocated to user $m$  in BS $f$ and otherwise $q^f_{m,c}=0$, and $\rho^f_{n,c}$  is the mapping between subcarriers and codebooks with $\rho^f_{n,c}=1$ if codebook $c$ consists of subcarrier $n$ in BS $f$ and otherwise $\rho^f_{n,c}=0$. We assume that $\boldsymbol{\rho}$ is a known parameter and  $p^f_{m,c}$ is assigned to subcarrier $n$ in codebook $c$ based on  $\eta^f_{n,c}$ with $0\le\eta^f_{n,c}\le 1$. This is determined based on the codebook design and satisfies $\sum_{\forall n \in \mathcal{N}_c}\eta^f_{n,c}=1,\,\,\forall c$. Note that $\mathcal{N}_c$ shows the subcarriers set of codebook $c$. Since we assume that each codebook is used only once  in each cell, we must have $\sum_{m\in\mathcal{M}_f}q^f_{m,c}=1,\,\,\forall c,f$.

\subsection{SCMA Signal Model: Received Signal}
Considering the parameters that are defined in the system model subsection, in an SCMA based HetNet system, the received signal at user $m$ in BS $f$ over subcarrier  $n$ is formulated as:
\begin{align}\label{SCMAsignal}
&Y_{m,n}^f=\sum_{i\in\mathcal{F}}\sum_{j\in\mathcal{M}_i}\sum_{c\in\mathcal{C}}q^i_{j,c} h^i_{m,n}\rho^i_{n,c}\eta^i_{n,c} \sqrt{p^i_{j,c}}s^i_{n,c}+w_{m,n}^f,
\end{align}
where $s^i_{n,c}$ is the message which BS $f$ sends on subcarrier $n$ in codebook $c$ and $w_{m,n}^f$ indicates the noise that user $m$ experiences on subcarrier  $n$.

 We note that over each subcarrier, each user receives the signals of all other users which use codebooks that contain subcarrier $n$.

\subsection{SCMA Signal Model: Receiver Side}
At the receiver side of user $m$ that uses code $c$, the MPA is run. The algorithm naturally cancels the interference from all other codes different from code $c$. Note that since code $c$ is reused, we have interference from code $c$ used in other cells over the considered user $m$ in cell $f$. With parameters defined in the system model subsection, in an SCMA based HetNet system, the SINR of user $m$ on codebook  $c$ over BS $f$ is given by \eqref{SCMAmulticellSINRmc} shown at the top of the next page. 
\begin{figure*}
\begin{align}\label{SCMAmulticellSINRmc}
&\gamma_{m,c}^f=\dfrac{q_{m,c}^fp_{m,c}^f\sum_{n\in \mathcal{N}}\eta_{n,c}^f\rho_{n,c}^f|h_{m,n}^f|^2}{\sum_{f'\in\mathcal{F}/\{f\}}\sum_{m'\in \mathcal{M}_{f'}}\sum_{n\in \mathcal{N}} q^{f'}_{m',c}p^{f'}_{m',c}\rho_{n,c}^{f'}\eta^{f'}_{n,c}|h^{f'}_{m,n}|^2+|w^f_{m,c}|^2}.
\end{align}
\end{figure*}

\section{Multiple Access Technique: PSMA}\label{Multiple Access Technique: PSMA}
In the proposed PSMA system, we assume that each codebook can be assigned to more than one user in a BS simultaneously. Based on this new approach, each codebook  is  assigned to more than one user by applying the SC method. On the receiver side, the users' signals are detected by using MPA and SIC. Based on the PSMA approach, each user can detect and remove the signals of users with worse average channel gain by applying MPA and SIC while considering  the signals of the users with better average channel gain as noise.
The transmitter and receiver block diagrams of the proposed PSMA-based system are shown in Fig. \ref{PSMA}. 
\begin{figure*}[t]
\centering
\includegraphics[width=1\textwidth]{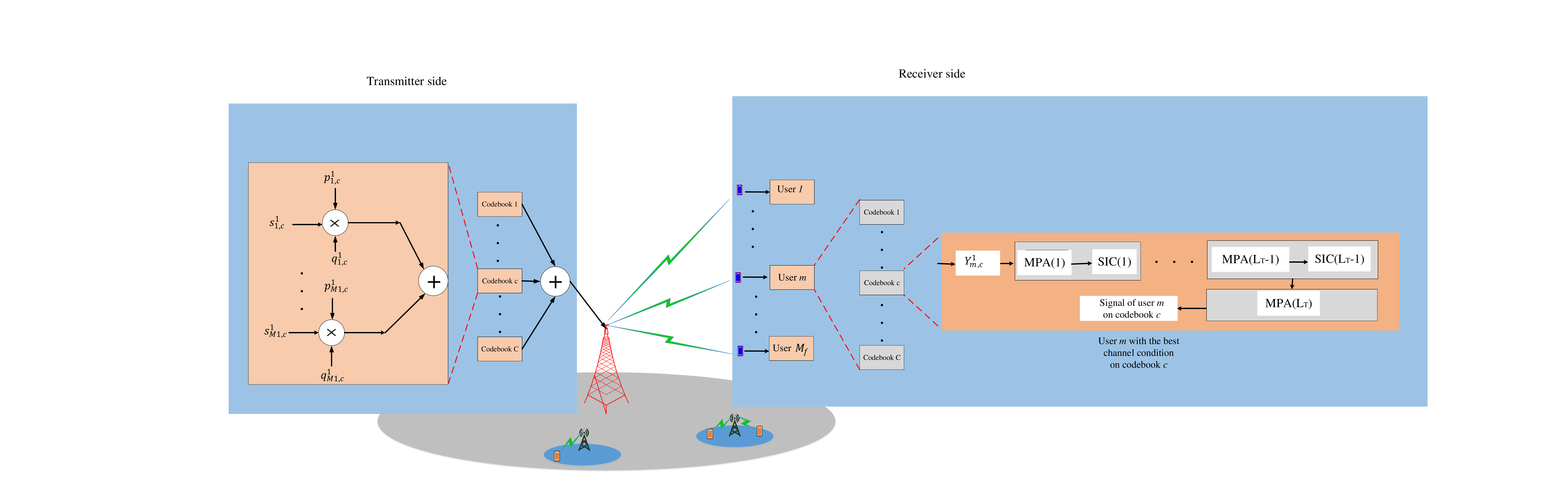}
\caption{Block diagram of transmitter and receiver in a PSMA-based system.}
\label{PSMA}
\end{figure*}

\subsection{PSMA System Model}

\subsubsection{PSMA-based single cell system model}
We consider  a downlink PSMA-based system with one BS, $M$ users, $N$ subcarriers, and $C$ codebooks. In this system model,  $h_{m,n}$ is the channel coefficient between user $m$ and BS on subcarrier $n$, $p_{m,c}$ is the transmit power of BS  to user $m$ on codebook $c$, and $q_{m,c}$ demonstrates the codebook assignment between user
$m$ and codebook $c$  with $q_{m,c}=1$ if codebook $c$
is allocated to user $m$ and otherwise $q_{m,c}=0$. In addition, $\rho_{n,c}$ indicates the mapping between subcarriers and codebooks with $\rho_{n,c}=1$ if codebook $c$ consists of subcarrier $n$, otherwise $\rho_{n,c}=0$. We assume that $\boldsymbol{\rho}$ is a known parameter.  Also, we assume that  $p_{m,c}$  is allocated to subcarrier $n$ in codebook $c$ based on  $\eta_{n,c}$ where $0\le\eta_{n,c}\le 1$ is determined based on the codebook design and satisfies $\sum_{\forall n \in \mathcal{N}_c}\eta_{n,c}=1\,\,\forall c$ \cite{NIC,int8}.

Note that in contrast to the SCMA model, a code can be reused by more than one user in the cell, i.e., for a code $c$, we may have $\sum_{m\in\mathcal{M}_f}q^f_{m,c}\geq 1$.

\subsubsection{PSMA-Based HetNet System Model}
We consider a downlink PSMA-based HetNet with $F$ BSs, $M$ users, and $C$ codebooks.  Here, we again emphasize that in contrast to the SCMA model, a code can be reused by more than one user in each cell.

\subsection{PSMA Signal Model: Received Signal}
In PSMA, a codebook can be reused more than once in each cell. The  received signal at user $m$ over subcarrier $n$ in cell $f$ (for a general case of multicell scenarios) is given by:
\begin{align}\label{PSMAsignalreceivedcell}
&Y_{m,n}^f= \sum_{f'\in\mathcal{F}}\sum_{m'\in\mathcal{M}_f}\sum_{c\in\mathcal{C}}q^{f'}_{m',c} h^{f'}_{m,n}\rho_{n,c}^{f'}\eta^{f'}_{n,c}\sqrt{p^f_{m',c}}s^{f'}_{n,c}\\&\nonumber+w_{m,n}^f,
\end{align}
where $s^i_{m,k}$ shows the message which is sent on subcarrier $k$ in codebook $c$.

 In SCMA, we assume the other signals coming from the different BSs  as noise. In PSMA, however, in addition to the different BSs signals, we accept some controllable interference to improve the system sum rate.  In SCMA, there is the constraint $\sum_{m\in \mathcal{M}_f}q^f_{m,c}\le 1,\,\,\,\forall c\in\mathcal{C},f\in\mathcal{F}$ which indicates that  in each BS, a codebook can be assigned to at most one user. However, in PSMA we have $\sum_{m\in \mathcal{M}_f}q^f_{m,c}\le L_T,\,\,\,\forall c\in\mathcal{C},f\in\mathcal{F}$  which shows  that  in each BS a codebook can be assigned to $L_T$ users simultaneously.

\subsection{PSMA Signal Model: Receiver Side}
Since the PSMA takes advantage of both PD-NOMA and SCMA techniques, the receiver performs the detection techniques of these schemes in an appropriate way. In the PSMA method, on the transmitter  side,  users are sorted based on certain criteria. On the receiver side, a user detects the signals of all other users that are worse than itself according to this criteria and removes them from the received signal using SIC. This user treats the signal of all users better than itself as noise. Note that in PSMA, to detect the signals of users, the message passing algorithm is run. Therefore, users with different codebooks do not interfere with each other and only users with the same codebook produce interference over each other. Therefore, by sorting users we mean sorting users that are using the same codebook. In this paper, the criterion for sorting users is their average channel gain in a codebook and given by
\begin{align}\label{hhatpsmareceiverside}
\hat{h}_{m,c}=\dfrac{\sum_{n=1}^{N} \rho_{n,c}|h_{m,n}|^2}{\sum_{n=1}^{N} \rho_{n,c}},
\end{align}
where we note that for each code $c$, we have $\sum_{n=1}^{N} \rho_{n,c}=U$. Also note that for two users $m$ and $m'$ using the same codebook $c$, we say user $m$ is better than user $m'$ if  we have $\hat{h}_{m,c}\geq\hat{h}_{m',c}$.

On the receiver side, user $m$, which is using codebook $c$, tries to decode the signals of all users using the same codebook $c$ with worse average channel gains. Assume  we have a single cell PSMA network and that user $m$ is better than user $m'$. Assuming that user $m$ has decoded the signals of all users worse than user $m'$ and subtracted them from the received signal, the available signal is given by 
\begin{align}
&Y_{m,c}=\sum_{n\in\mathcal{N}} h_{m,n}\rho_{n,c}\eta_{n,c} \sqrt{p_{m',c}}s_{n,c}\\\nonumber&+     \sum_{j\in \mathcal{M},|\hat{h}_{m,c}|^2\ge|\hat{h}_{m',c}|^2}\sum_{n\in\mathcal{N}}q_{j,c} h_{m,n}\rho_{n,c}\eta_{n,c} \sqrt{p^f_{j,c}}s_{n,c}+w_{m,c},
\end{align}
and the SINR of user $m'$ at user $m$ is given by
\begin{align}\label{sinrlema1mmprinscmasinglecell}
& \gamma_{m',c}(m)=\\&\nonumber \dfrac{q_{m,c}p_{m',c}\sum_{n\in \mathcal{N}}\eta_{n,c}\rho_{n,c}|h_{m,n}|^2}{\sum_{i\in \mathcal{M},\hat{h}_{i,c}\ge\hat{h}_{m',c} }q_{i,c}p_{i,c}\sum_{n\in\mathcal{N}}\eta_{n,c}\rho_{n,c}|h_{m,n}|^2+|w_{m,c}|^2}
\end{align}
Note that the SINR of user $m'$ at its receiver is given by:
\begin{align}\label{sinrlema1mprinmprinscmasinglecell}
& \gamma_{m',c}(m')=\\&\nonumber \dfrac{q_{m',c}p_{m',c}\sum_{n\in \mathcal{N}}\eta_{n,c}\rho_{n,c}|h_{m',n}|^2}{\sum_{i\in \mathcal{M},\hat{h}_{i,c}\ge\hat{h}_{m',c} }q_{i,c}p_{i,c}\sum_{n\in\mathcal{N}}\eta_{n,c}\rho_{n,c}|h_{m',n}|^2+|w_{m,c}|^2}.
\end{align}
We emphasize that for user $m$ to correctly decode the signal of user $m'$, we must have $ \gamma_{m',c}(m)\geq  \gamma_{m',c}(m')$. For various power allocation strategies for the codebooks, ensuring $\hat{h}_{m,c}\geq\hat{h}_{m',c}$ does not guarantee the requirement $ \gamma_{m',c}(m)\geq  \gamma_{m',c}(m)$. Therefore, in some cases, we must explicitly declare this condition as an optimization constraint.

Subtracting the detected signals of all users worse than user $m$ from the received signal, 
 the corresponding SINR is given by:
 \begin{align}\label{snr}
& \gamma_{m,c}=\\&\nonumber \dfrac{q_{m,c}p_{m,c}\sum_{n\in \mathcal{N}}\eta_{n,c}\rho_{n,c}|h_{m,n}|^2}{\sum_{i\in \mathcal{M},\hat{h}_{i,c}\ge\hat{h}_{m,c} }q_{i,c}P_{i,c}\sum_{n\in\mathcal{N}}\eta_{n,c}\rho_{n,c}|h_{m,n}|^2+|w_{m,c}|^2}.\\&\nonumber
\end{align}

For multicell scenario, the same procedure is applied. However, the signals received from other cells are treated as noise. After detecting the signals of all users worse than user $m$, the available signal for user $m$ is given by:
\begin{align}\label{PSMAsignalreceivedcell0}
&Y_{m,c}^f=\sum_{k\in\mathcal{N}} h^f_{m,k}\rho_{n,c}^f\eta^f_{k,c} \sqrt{p^f_{m,c}}s^f_{m,k}\\\nonumber&+     \sum_{j\in \mathcal{M}_f,|\hat{h}^f_{j,c}|^2\ge|\hat{h}^f_{m,c}|^2}\sum_{k\in\mathcal{N}}q^f_{j,c} h^f_{m,k}\rho_{n,c}^f\eta^f_{k,c} \sqrt{p^f_{j,c}}s^f_{m,k} \\\nonumber&+\sum_{i\in\mathcal{F}, i\ne f}\sum_{j\in\mathcal{M}_f}\sum_{k\in\mathcal{N}}q^i_{j,c} h^i_{m,k}\rho_{k,c}^i\eta^i_{k,c} \sqrt{p^i_{j,c}}s^i_{m,k}+w_{m,c}^f,
\end{align}
 where the first term of \eqref{PSMAsignalreceivedcell0} represents the signal of user $m$ on codebook $c$ over BS $f$. The second term  shows the  interference that comes from non-orthogonality in power domain. The third term indicates inter-cell interference. The corresponding SINR is given by
 \begin{equation}\label{snrscm}
 \gamma^f_{m,c}=\dfrac{q^f_{n,c}\sum_{n\in \mathcal{N}}\rho_{n,c}^f\eta^f_{n,c}p^f_{m,c}|h^f_{m,n}|^2}{I^f_{m,c}(\text{Intercell})+I^f_{m,c}(\text{NOMA})+|w^f_{m,c}|^2},
\end{equation}
 where $I^f_{m,n}$  indicates the  intercell interference which is given by
\begin{align}
&I^f_{m,c}(\text{Intercell})=\sum_{f'\in\mathcal{F}/\{f\}}\sum_{m\in \mathcal{M}_{f'}}\sum_{n\in \mathcal{N}} q^{f'}_{m,c}p^{f'}_{m,c}\rho_{n,c}^{f'}\\\nonumber&\eta^{f'}_{n,c}|h^{f'}_{m,n}|^2,
\end{align}
and $I(\text{NOMA})$  shows the interference that comes from using the PD-NOMA technique and is given by
\begin{equation}
I^f_{m,c}(\text{NOMA})=\sum_{i\in \mathcal{M}_f,\hat{h}^f_{i,c}\ge\hat{h}^f_{m,c}}\sum_{n\in\mathcal{N}}q^f_{i,c}p^f_{i,c}\rho_{n,c}^f\eta^f_{n,c}|h^f_{m,n}|^2.
\end{equation}

Assuming that user $m$ at BS $f$ on codebook $c$  should detect and remove  the signal of $L_T-1$ users and defining the average channel gains set of superimposed users on codebook $c$  as $\hat{\mathcal{H}}_c$, the detection steps  are explained in Algorithm \eqref{table-20}.
\begin{algorithm}
\caption{PSMA based receiver }
\label{table-20}
I:  Set $b=1$ as iteration number,\\
II: $\hat{H}^f_{b,c}=\min{\hat{\mathcal{H}}_c}$
\\
III: Apply MPA on $Y(b)$,\\
 \hspace{1cm}Output: $\sum_{k\in\mathcal{N}}\eta^f_{k,c}s^f_{b,k}$ (signal of user $b$ on codebook $c$ over BS $f$)\\
V: Apply SIC:\\
$Y(b) = Y(b)-\sum_{k\in\mathcal{N}} p_{b,c}^f\eta^f_{k,c}s^f_{b,k}$,\\
 VI: If $b=L_T$\\
 \hspace{1cm} Apply MPA on $Y(b)$ as  signal of user $m$,\\
 \hspace{1cm} else\\
 \hspace{1cm}set $b=b+1$,\\
 \hspace{1cm} set $\hat{\mathcal{H}}_c=\hat{\mathcal{H}}_c-\{\hat{H}^f_{b,c}\}$,\\
  \hspace{1cm}  go back to  II.
\end{algorithm}

\section{Receiver Complexity}\label{Receiver Complexity}
In this section, the complexity of PD-NOMA, SCMA, and PSMA are investigated. The computational complexity of PD-NOMA and SCMA receiver  is investigated in \cite{our}. If we assume that in PSMA each codebook is assigned to $L_T$ users simultaneously, each user should be able to apply MPA $L_T$ times and SIC $L_T-1$ times to detect and decode the transmitted data (Fig. \ref{PSMA}). Therefore, if $G$ codebooks are assigned to a user,  the complexity order of  PSMA is approximately  given by
\begin{align}
\mathcal{O}\big((I_T(|\boldsymbol{\pi}|^d))(G)(L_T)\big),
\end{align}
where
$\mathbf{\pi}$ shows  the codebook set size, $I_T$ denotes the
total number of iterations, and $d$ represents the non-zero elements
in each row of the matrix $\bold{X}$ where
$\bold{X}=(\bold{x}_1,\dots,\bold{x}_n)$ is the factor graph matrix. 

The receiver complexity order for PD-NOMA, SCMA and PSMA are summarized in Fig. \ref{table-45}. By setting the system parameters as 
demonstrated in Fig. \ref{table-5}, we can see that the PSMA complexity increases around an order of magnitude with respect to SCMA. However, as will be shown in simulation results, this comes with significant improvement in system throughput.
In this table, $L_{T'}$ and $G'$ show the total number of users that can be assigned to each user and the total number of subcarriers  assigned to each user  in PD-NOMA, respectively.
\begin{figure*}[t]
                 \centering
                 \caption{Comparison between PD-NOMA, SCMA and PSMA receiver complexity }
                 \label{table-5}
                 \begin{tabular}{ |c|c|c|c|c|c|c|c|c|c|}
                 \hline

                     N & d&U &$L_{T'}$&$G'$& $L_T$&G& NOMA-complexity & SCMA-complexity& PSMA-complexity\\
                 \hline

                 8 & 3& 2&3&4&3 & 4&360&1536&12*1536\\

                 \hline

                  10& 4&3 &4&5&4&5&1920&40000&20*40000\\

                 \hline
                 \end{tabular}
\end{figure*}
\begin{figure*}[t]
	\centering
	\caption{Receiver complexity order for PD-NOMA, SCMA and PSMA }
	\label{table-45}
	\begin{tabular}{ |c|c|c|}
		\hline
		
		NOMA-complexity order & SCMA-complexity order& PSMA-complexity order\\
		\hline
		$\mathcal{O}(\big(2L_{T'}^3)+(2L_{T'}^2)(G')\big)(L_{T'}-1))$&$\mathcal{O}\big((I_T(|\boldsymbol{\pi}|^d))\big)$&$\mathcal{O}\big((I_T(|\boldsymbol{\pi}|^d))(G)(L_T)\big)$\\
		
		\hline
		
	\end{tabular}
\end{figure*}
\section{Resource Allocation Problems for PSMA based Cellular Networks}\label{Optimization Problems for PSMA based Cellular Networks}
In this section, we consider resource allocation problems for the proposed PSMA framework to assess the performance  compared to SCMA and PD-NOMA.
To get a better insight, we first consider the simple case of a single cell system with equal proportion of power for each subcarrier. Then the unequal
 proportion of power is considered. Finally, a comprehensive HetNet system is analyzed.

\subsection{PSMA-based Single Cell System with Equal Proportion of Power for Each Subcarrier}
Here, we assume that $\eta_{n,c}$ for each subcarrier in codebook $c$ has equal value. We use the the following Remark from \cite{lem1} in Chapter 6 to propose
Lemma 1:

\begin{rem}\label{lema1}
The interference cancellation (IC) of a candidate user (worse user) at user $m$ (better user) can be successfully performed if the (corresponding) SINR of the candidate user with less average channel gain, which is measured at user $m$, is more than its SINR on its receiver  (\cite{lem1} chapter 6). For more clarification, suppose $m$ and $m'$ are two users with $\hat{h}_{m,c}\ge \hat{h}_{m',c}$. To achieve successful IC at user $m$, we should have $\gamma_{m',c}(m)\ge q_{m,c}\gamma_{m',c}(m')$ where $\gamma_{m',c}(m)$ is the SINR of user $m'$ at user $m$ on codebook $c$ and $\gamma_{m',c}(m)$ and $\gamma_{m',c}(m')$ are, respectively, given by \eqref{sinrlema1mmprinscmasinglecell} and \eqref{sinrlema1mprinmprinscmasinglecell}.
\end{rem}

\begin{theorem}
With uniform  power allocation, ($\eta_{n,c}=1/U, \forall n\in \mathcal{N}_c$), the SIC constraint of Remark 1 is always satisfied.
\end{theorem}
\begin{proof}
For $\hat{h}_{m,c}>\hat{h}_{m',c}$,  we can write $\gamma_{m',c}(m)$ and $\gamma_{m',c}(m')$ as follows:

 \begin{align}\label{snr1}
& \gamma_{m',c}(m)=\\&\nonumber \dfrac{q_{m,c}p_{m',c}\sum_{n\in \mathcal{N}}\eta_{n,c}\rho_{n,c}|h_{m,n}|^2}{\sum_{i\in \mathcal{M},\hat{h}_{i,c}>\hat{h}_{m',c} }q_{i,c}p_{i,c}\sum_{n\in\mathcal{N}}\eta_{n,c}\rho_{n,c}|h_{m,n}|^2+|w_{m,c}|^2}\\&\nonumber =\dfrac{q_{m,c}1/U p_{m',c}\sum_{n\in \mathcal{N}}\rho_{n,c}|h_{m,n}|^2}{\sum_{i\in \mathcal{M},\hat{h}_{i,c}>\hat{h}_{m,c} }q_{i,c}1/U p_{i,c}\sum_{n\in\mathcal{N}}\rho_{n,c}|h_{m,n}|^2+|w_{m,c}|^2}\\&\nonumber =\dfrac{q_{m,c}1/U p_{m',c}}{\sum_{i\in \mathcal{M},\hat{h}_{i,c}>\hat{h}_{m,c} }q_{i,c}1/U p_{i,c}+|w'_{m,c}|^2},
\end{align}

 \begin{align}
& \gamma_{m',c}(m')=\\&\nonumber \dfrac{q_{m',c}p_{m',c}\sum_{n\in \mathcal{N}}\eta_{n,c}\rho_{n,c}|h_{m',n}|^2}{\sum_{i\in \mathcal{M},\hat{h}_{i,c}>\hat{h}_{m',c} }q_{i,c}p_{i,c}\sum_{n\in\mathcal{N}}\eta_{n,c}\rho_{n,c}|h_{m',n}|^2+|w_{m,c}|^2}\\&\nonumber =\dfrac{q_{m',c}1/U p_{m',c}\sum_{n\in \mathcal{N}}\rho_{n,c}|h_{m',n}|^2}{\sum_{i\in \mathcal{M},\hat{h}_{i,c}>\hat{h}_{m',c} }q_{i,c}1/U p_{i,c}\sum_{n\in\mathcal{N}}\rho_{n,c}|h_{m',n}|^2+|w_{m,c}|^2}\\&\nonumber =\dfrac{q_{m',c}1/U p_{m',c}}{\sum_{i\in \mathcal{M},\hat{h}_{i,c}>\hat{h}_{m',c} }1/U p_{i,c}q_{i,c}+|w''_{m,c}|^2},
\end{align}
where $|w'_{m,c}|^2$ and $|w''_{m,c}|^2$  are, respectively, given by

$$|w'_{m,c}|^2=\dfrac{|w_{m,c}|^2}{\sum_{n\in \mathcal{N}}\rho_{n,c}|h_{m,n}|^2},$$
and
$$|w''_{m,c}|^2=\dfrac{|w_{m,c}|^2}{\sum_{n\in \mathcal{N}}\rho_{n,c}|h_{m',n}|^2}.$$
As can be seen from \eqref{hhatpsmareceiverside} and $\hat{h}_{m,c}\ge\hat{h}_{m',c}$,  we can conclude that
 $\gamma_{m',c}(m)\ge q_{m,c}\gamma_{m',c}(m')$.
\end{proof}

To formulate the optimization problem, we note that the rate of user $m$ on codebook $c$  is given by
 \begin{equation}
 r_{m,c}(\mathbf{P},\mathbf{Q})=\log(1+\gamma_{m,c}),
\end{equation}
where $\gamma_{m,c}$  indicates the SINR of user $m$ on codebook $c$ which is given by \eqref{snr}, $\boldsymbol{Q}'=\big[q_{m,c}\big]\,\, \forall m\in \mathcal{M}, c\in \mathcal{C}$ and  $\boldsymbol{P}'=\big[p_{m,c}\big]\,\, \forall m\in \mathcal{M}, c\in \mathcal{C}$.

The problem formulation of this system model is formulated as follows:
\begin{subequations}\label{orj_p_sc1}
\begin{align}\label{eeq8a1}
&\max_{\mathbf{Q}',\mathbf{P}'}\; \sum_{m\in \mathcal{M}}\sum_{c\in \mathcal{C}}r_{m,c}(\mathbf{P}',\mathbf{Q}'),\\& \label{eeq8b1}
 \text{s.t.}:\hspace{.25cm}
\sum_{m\in \mathcal{M}}\sum_{c\in \mathcal{C}}q_{m,c}p_{m,c}\le p_{\text{max} },\\&\label{eeq8c1}
\hspace{1cm}\sum_{m\in \mathcal{M}}\sum_{m\in \mathcal{C}}q_{m,c}\rho_{n,c}\le K,\,\,\,\forall n\in\mathcal{N},\\&\label{eeq8f1}
\hspace{1cm}\sum_{m\in \mathcal{M}}q_{m,c}\le L_T,\,\,\,\forall c\in\mathcal{C},\\&\label{eeq8h1}
\hspace{1cm}p_{m,c}\ge 0,\,\,\, \forall m\in \mathcal{M} ,c\in\mathcal{C},\\&\label{eeq8e1}
\hspace{1cm}q_{m,c}\in
\begin{Bmatrix}
 0 ,
1
\end{Bmatrix},\,\,\forall m\in \mathcal{M} ,n\in\mathcal{N},
\end{align}
\end{subequations}
where \eqref{eeq8b1} indicates the maximum available transmit power in  BS, \eqref{eeq8c1} shows that each subcarrier can be reused at most $K$ times, and \eqref{eeq8f1} demonstrates that each codebook can be assigned to $L_T$ users  simultaneously.

\subsection{PSMA Based Single Cell System with  Unequal Proportion of Power  for Each Subcarrier}
If we consider unequal power proportion for different subcarriers in a codebook, to achieve successful IC, we cannot use the concept  used in the previous system model. Based on Remark \ref{lema1},  to achieve successful IC, in this system model,  the following constraint should be  applied:
\begin{equation}\label{gamaorder}
\gamma_{i,c}(m)\ge q_{m,c}\gamma_{i,c}(i) \,\,\,\, \forall i,m\in{\mathcal{M}},\hat{h}_{i,c}\ge \hat{h}_{m,c}.
\end{equation}

The resource allocation problem is formulated as:
\begin{subequations}\label{orj_p_sc11}
\begin{align}\label{eeq8a11}
&\max_{\mathbf{Q},\mathbf{P}}\; \sum_{m\in \mathcal{M}}\sum_{c\in \mathcal{C}}r_{m,c}(\mathbf{P},\mathbf{Q}),\\& \nonumber
 \text{s.t.}:\hspace{.25cm}
\eqref{eeq8b1}-\eqref{eeq8e1},\\&
\hspace{1cm}\gamma_{i,c}(m)\ge q_{m,c}\gamma_{i,c}(i) \,\,\forall i,m\in{\mathcal{M}},\hat{h}_{i,c}>\hat{h}_{m,c}.
\end{align}
\end{subequations}
\subsection{PSMA-based HetNet System Model}

Considering Remark \ref{lema1},  to achieve successful IC,  the following constraint is applied:

\begin{equation}\label{gamaorder1}
\gamma^f_{m,c}(j)\ge q^f_{j,c}\gamma^f_{m,c}(m) \,\,\,\, \forall j,m\in{\mathcal{M}_f}, f\in{\mathcal{F}},\hat{h}^f_{m,c}>\hat{h}^f_{j,c},
\end{equation}
where  $\gamma^f_{i,c}(m)$ shows the SINR of user $i$ at user $m$ in codebook $c$ over BS $f$.

The rate of user $m$ on codebook $c$ over BS $f$ is given by
 \begin{equation}
 r^f_{m,c}(\mathbf{P},\mathbf{Q})=\log(1+\gamma^f_{m,c}),
\end{equation}
where $\gamma^f_{m,c}$  indicates the SINR of user $m$ on codebook $c$ over BS $f$ which is given by \eqref{snrscm}, $\boldsymbol{Q}=\big[q^f_{m,c}\big]\,\, \forall m\in \mathcal{M}_f, f\in \mathcal{F}, c\in \mathcal{C}$ and  $\boldsymbol{P}=\big[p^f_{m,c}\big]\,\, \forall m\in \mathcal{M}_f, f\in \mathcal{F}, c\in \mathcal{C}$.

The proposed optimization problem formulation of the PSMA based HetNet system is written as:
\begin{subequations}\label{orj_p_sc}
\begin{align}\label{eeq8a}
&\max_{\mathbf{Q},\mathbf{P}}\; \sum_{f\in\mathcal{F}}\sum_{m\in \mathcal{M}_f}\sum_{c\in \mathcal{C}}r^f_{m,c}(\mathbf{P},\mathbf{Q}),\\& \label{eeq8b}
 \text{s.t.}:\hspace{.25cm}
\sum_{m\in \mathcal{M}_f}\sum_{c\in \mathcal{C}}q^f_{m,c}p^f_{m,c}\le p^f_{\text{max} }\,\,\,\forall f\in\mathcal{F},\\&\label{eeq8c}
\hspace{1cm}\sum_{m\in \mathcal{M}_f}\sum_{c\in \mathcal{C}}q^f_{m,c}\rho^f_{n,c}\le K,\,\,\,\forall n\in\mathcal{N},f\in\mathcal{F},\\&\label{eeq8f}
\hspace{1cm}\sum_{m\in \mathcal{M}_f}q^f_{m,c}\le L_T,\,\,\,\forall c\in\mathcal{C},f\in\mathcal{F},\\&\label{eeq8h}
\hspace{1cm}\gamma^f_{m,c}(j)\ge q^f_{j,c}\gamma^f_{m,c}(m), \,\forall j,m\in{\mathcal{M}_f}, f\in{\mathcal{F}},\\&\nonumber
\hspace{1cm}\hat{h}^f_{m,c}>\hat{h}^f_{j,c},\\&\label{eeq8ee}
\hspace{1cm}p^f_{m,c}\ge 0,\,\,\, \forall m\in \mathcal{M}_f ,c\in\mathcal{C},f\in\mathcal{F},\\&\label{eeq8e}
\hspace{1cm}q^f_{m,c}\in
\begin{Bmatrix}
 0 ,
1
\end{Bmatrix},\,\,\forall m\in \mathcal{M}_f ,n\in\mathcal{N},f\in\mathcal{F},
\end{align}
\end{subequations}
where \eqref{eeq8b} indicates the maximum available transmit power in each BS, \eqref{eeq8c} shows that each subcarrier can be reused at most $K$ times, and \eqref{eeq8f} demonstrates that each codebook can be assigned to $L_T$ users  simultaneously in each BS.
\begin{rem}
Constraint  \eqref{eeq8h} is a linear constraint in terms of power variables.
To show the linearity of constraint  \eqref{eeq8h}, we obtain equation \eqref{lineprof} shown at the top of the next page
\begin{figure*}[t]
\begin{align}\label{lineprof}
&\gamma^f_{m,c}(j)\ge q^f_{j,c}\gamma^f_{m,c}(m)\Rightarrow\\&\nonumber
\dfrac{q^f_{j,c}\sum_{n\in \mathcal{N}}\rho_{n,c}^f\eta^f_{n,c}p^f_{m,c}|h^f_{j,n}|^2}{\sum_{f'\in\mathcal{F}/\{f\}}\sum_{m\in \mathcal{M}_{f'}}\sum_{n\in \mathcal{N}} q^{f'}_{n,c}p^{f'}_{m,c}\rho_{n,c}^{f'}\eta^{f'}_{n,c}|h^{f'}_{j,n}|^2+\sum_{i\in \mathcal{M}_f,\hat{h}^f_{i,c}>\hat{h}^f_{m,c}}\sum_{n\in\mathcal{N}}q^f_{i,c}p^f_{i,c}\rho_{n,c}^f\eta^f_{n,c}|h^f_{j,n}|^2+|w^f_{j,c}|^2} \\&\nonumber
\dfrac{q^f_{j,c}q^f_{m,c}\sum_{n\in \mathcal{N}}\eta^f_{n,c}p^f_{m,c}|h^f_{m,n}|^2}{\sum_{f'\in\mathcal{F}/\{f\}}\sum_{m\in \mathcal{M}_{f'}}\sum_{n\in \mathcal{N}} q^{f'}_{n,c}p^{f'}_{m,c}\rho_{n,c}^{f'}\eta^{f'}_{n,c}|h^{f'}_{m,n}|^2+\sum_{i\in \mathcal{M}_f,\hat{h}^f_{i,c}>\hat{h}^f_{m,c}}\sum_{n\in\mathcal{N}}q^f_{i,c}p^f_{i,c}\rho_{n,c}^f\eta^f_{n,c}|h^f_{m,n}|^2+|w^f_{m,c}|^2},
\end{align}
\hrule
\end{figure*}
and after simplifying, the linear constraint  is obtained by \eqref{rem1} shown at the top of the next page.
\begin{figure*}[t]
\begin{align}\label{rem1}
&-(q^f_{j,c}q^f_{m,c}\sum_{n\in \mathcal{N}}\eta^f_{n,c}|h^f_{m,n}|^2)(\sum_{f'\in\mathcal{F}/\{f\}}\sum_{m\in \mathcal{M}_{f'}}\sum_{n\in \mathcal{N}} q^{f'}_{n,c}p^{f'}_{m,c}\rho_{n,c}^{f'}\eta^{f'}_{n,c}|h^{f'}_{j,n}|^2+\sum_{i\in \mathcal{M}_f,\hat{h}^f_{i,c}>\hat{h}^f_{m,c}}\sum_{n\in\mathcal{N}}q^f_{i,c}p^f_{i,c}\rho_{n,c}^f\eta^f_{n,c}|h^f_{j,n}|^2+|w^f_{j,c}|^2) \\&\nonumber+(q^f_{m,c}\sum_{n\in \mathcal{N}}\rho_{n,c}^f\eta^f_{n,c}|h^f_{j,n}|^2)(\sum_{f'\in\mathcal{F}/\{f\}}\sum_{m\in \mathcal{M}_{f'}}\sum_{n\in \mathcal{N}} q^{f'}_{n,c}p^{f'}_{m,c}\rho_{n,c}^{f'}\eta^{f'}_{n,c}|h^{f'}_{m,n}|^2+\sum_{i\in \mathcal{M}_f,\hat{h}^f_{i,c}>\hat{h}^f_{m,c}}\sum_{n\in\mathcal{N}}q^f_{i,c}p^f_{i,c}\rho_{n,c}^f\eta^f_{n,c}|h^f_{m,n}|^2+|w^f_{m,c}|^2)<0.
\end{align}
\hrule
\end{figure*}

\end{rem}
It should be mentioned that, the problem formulation and solution methods of SCMA and PD-NOMA based systems are investigated in \cite{our}.

\section{ solution algorithm of proposed resource allocation problems}\label{solution of the joint power and subcarrier allocation}
We propose the solution for the problem corresponding to the PSMA-based HetNet which is the most general one. The solution for the 
the first two problems can be considered as  special cases.
\subsection{Solution Algorithm}
Problem \eqref{orj_p_sc} is non-convex and  contains both integer and continuous variables. Therefore, the available methods to solve convex problems cannot be applied directly.
To solve the proposed problem, an iterative algorithm based on the SCA method is applied and power and codebook are assigned  separately in each iteration. The power allocation problem is non-convex, and the codebook allocation problem is integer non-linear programing (INLP). To solve the power allocation  problem first, successive convex approximation with low complexity (SCALE) \cite{SCALE} and difference of two concave functions (DC) \cite{DCDC} is applied to approximate the problem by a convex one, then the dual method is used.  To solve the codebook assignment problem, the mesh adaptive direct search (MADS) \cite{bib36} algorithm is exploited. An overview of the iterative algorithm is shown in Algorithm \ref{table-1}.
\begin{algorithm}
\caption{Overview of the solution algorithm }
\label{table-1}

I: Initialize  $\mathbf{Q}(0)$, $\mathbf{P}(0)$ and set $t=0$ (iteration number).
\\
II: Repeat:
\\
III: Set  $\mathbf{Q}=\mathbf{Q}(t)$ and find a solution for problem \eqref{orj_p_sc} by applying the SCA approach and
assign it to $\mathbf{P}(t+1)$,\\
IV: Find $\mathbf{Q}(t+1)$ by solving \eqref{orj_p_sc}
with  $\mathbf{P}=\mathbf{P}(t+1)$,\\
V: When $\|\mathbf{P}(t)-\mathbf{P}(t-1)\|\le \Upsilon$
stop. \\Otherwise,\\
 set $t=t+1$ and go back to III.\\
 Output:
 \\
 $\boldsymbol{\rho}(k)$ and $\mathbf{P}(k)$ are adopted for the considered system.
\end{algorithm}
\subsection{Power Allocation }
The power allocation problem is written as:
\begin{subequations}\label{orj_p_sc_p}
\begin{align}
&\max_{\mathbf{P}}\; \sum_{f\in\mathcal{F}}\sum_{m\in \mathcal{M}_f}\sum_{c\in \mathcal{C}}r^f_{m,c}(\mathbf{P}),\\& \nonumber
 \text{s.t.}:\hspace{.25cm}
\eqref{eeq8b},\eqref{eeq8h}.\\&\nonumber
\end{align}
\end{subequations}
In order to tackle the non-convexity issue of this problem, the following inequality (known as SCALE inequality) is applied  \cite{SCALE}:
\begin{equation}\label{eq11}
\xi \,\log(\mathcal{X})+\psi\le \,\log(1+\mathcal{X}),
\end{equation}
where
$$\xi=\dfrac{\mathcal{X}_0}{\mathcal{X}_0+1},~\psi=\log(1+\mathcal{X}_0)-\dfrac{\mathcal{X}_0}{\mathcal{X}_0+1}\,\log(\mathcal{X}_0).$$

By exploiting this inequality, the objective function of problem \eqref{orj_p_sc_p} is written by:
$$\sum_{f\in\mathcal{F}}\sum_{m\in \mathcal{M}_f}\sum_{c\in \mathcal{C}}\xi^f_{m,c}\log(\gamma^f_{m,c})+\eta^f_{m,c}.$$
The objective function is still non-convex. By transforming $p^f_{m,c}=\exp(\tilde{p}^f_{m,c})$, the convex form of the objective function is achieved. However, using the mentioned transformation, constraint \eqref{eeq8h} becomes a non-convex constraint. To tackle this issue, an approximation method such as DC \cite{DCDC} is applied to approximate \eqref{eeq8h} by a convex function. Therefore, based on DC, approximation of constraint  \eqref{eeq8h} is written as \eqref{DCrem1}.
\begin{figure*}[t]
\begin{subequations}
\begin{align}\label{DCrem1}
&-(\sum_{n\in \mathcal{N}}\eta^f_{n,c}|h^f_{m,n}|^2)(\sum_{f'\in\mathcal{F}/\{f\}}\sum_{m'\in \mathcal{M}_{f'}}\sum_{n\in \mathcal{N}} q^{f'}_{n,c}(\big(     {\exp{\tilde{p}_{m',c}^{f',t-1}}}+ {\exp{\tilde{p}_{m',c}^{f',t-1}}}(\tilde{p}_{m',c}^{f'}-\tilde{p}_{m',c}^{f',t-1})\big))\rho_{n,c}^{f'}\eta^{f'}_{n,c}|h^{f'}_{j,n}|^2\\&\nonumber+\sum_{i\in \mathcal{M}_f,|\hat{h}^f_{j,c}|^2>|\hat{h}^f_{m,c}|^2}\sum_{n\in\mathcal{N}}q^f_{i,c}\big(     {\exp{\tilde{p}_{i,c}^{f,t-1}}}+ {\exp{\tilde{p}_{i,c}^{f,t-1}}}(\tilde{p}_{i,c}^{f}-\tilde{p}_{i,c}^{f,t-1})\big)\rho_{n,c}^f\eta^f_{n,c}|h^f_{j,n}|^2+|w^f_{j,c}|^2)+(\sum_{n\in \mathcal{N}}\rho_{n,c}^f\eta^f_{n,c}|h^f_{j,n}|^2) \\&\nonumber(\sum_{f'\in\mathcal{F}/\{f\}}\sum_{m'\in \mathcal{M}_{f'}}\sum_{n\in \mathcal{N}} q^{f'}_{n,c}p^{f'}_{m',c}\rho_{n,c}^{f'}\eta^{f'}_{n,c}|h^{f'}_{m',n}|^2+\sum_{i\in \mathcal{M}_f,|\hat{h}^f_{j,c}|^2>|\hat{h}^f_{m,c}|^2}\sum_{n\in\mathcal{N}}q^f_{i,c}p^f_{i,c}\rho_{n,c}^f\eta^f_{n,c}|h^f_{m,n}|^2+|w^f_{m,c}|^2)<0.
\end{align}
\end{subequations}
\hrule
\end{figure*}

Consequently, a convex optimization problem in standard form  with variables $\tilde{\mathbf{P}}$ is achieved as follows:
\begin{align}\label{poweral}
&\max_{\tilde{\mathbf{P}}}\; \sum_{f\in\mathcal{F}}\sum_{m\in \mathcal{M}_f}\sum_{c\in \mathcal{C}}r^f_{m,c}(\tilde{p}^f_{m,c}),\\& \nonumber
 \text{s.t.}:\hspace{.25cm}
\sum_{m\in \mathcal{M}_f}\sum_{n\in \mathcal{N}}q^f_{m,c}\exp(\tilde{p}^f_{m,c})\le p^f_{\text{max} }\,\,\,\forall f\in\mathcal{F},\\&\nonumber
  \hspace{1cm}   \eqref{DCrem1}.
\end{align}
To show the concavity of the objective function, we rewrite it as follows:
\begin{align}\label{p.o.c}
&\sum_{f\in \mathcal{F}}\sum_{m\in \mathcal{M}_f} \sum_{c\in \mathcal{C}}\xi^f_{m,c}\bigg(\log(q^f_{n,c}\sum_{n\in \mathcal{N}}\rho^f_{n,c}\eta^f_{n,c}|h^f_{m,n}|^2)+\tilde{p}^f_{m,c}\\\nonumber& - \log\bigg(\sum_{f\in \mathcal{F}/\{f\}}\sum_{{m\in \mathcal{M}_f}} q^f_{n,c}\sum_{n\in \mathcal{N}}\rho^f_{n,c}\eta^f_{n,c}\exp(\tilde{p}^f_{m,c})|h^{f}_{m,n}|^2+\\\nonumber&\sum_{i\in \mathcal{M}_f,|\hat{h}^f_{i,c}|^2>|\hat{h}^f_{m,c}|^2}\sum_{n\in\mathcal{N}}q^f_{i,c}\exp(\tilde{p}^f_{i,c})\rho^f_{n,c}\eta^f_{n,c}|h^f_{m,n}|^2\\\nonumber&+|w^{f}_{m,c}|^2\bigg)\bigg)+\psi ^f_{m,c}.
\end{align}
Each term in \eqref{p.o.c} is concave and therefore, the new objective function is concave. We note that the log-sum-exp function is convex \cite{boyd}.

To deploy the SCALE algorithm in Algorithm \ref{table-1}, we use an iterative power allocation algorithm to find a power allocation not worse than $\bold{P}^t$ when $\boldsymbol{Q} = \boldsymbol{Q}^t$. This procedure is presented as in Algorithm \ref{table-2} where $z$  indicates the iteration number and $\bold{P}^{t,z}$ shows the power allocation in iteration $z$. In each iteration $t$, the power allocation problem is solved. Moreover,  $\boldsymbol{\xi}^z$ and $\boldsymbol{\psi}^z$ are updated as $\boldsymbol{\xi}^{z+1}$ and $\boldsymbol{\psi}^{z+1}$, respectively. The algorithm is initialized by $\boldsymbol{\xi}^{0}=1$ and $\boldsymbol{\psi}^{0}=0$, and continued until $\|\bold{P}^{t,z}-\bold{P}^{t,z-1}\|\le \epsilon$.

\begin{algorithm}
\caption{ALGORITHM TO IMPROVE THE SIC METHOD }
\label{table-2}
I:  Set $z=0$ and initialize  $\boldsymbol{\xi}^{0}=1$ and $\boldsymbol{\psi}^{0}=0$,\\
II: Repeat:
\\
III: Solve \eqref{poweral} then give the solution to $\bold{P}^{t,z}$, \\
IV:Update $\boldsymbol{\xi}^{z+1}$ and $\boldsymbol{\psi}^{z+1}$ with $\bold{P}^{t,z}$, \\
V: When $\|\bold{P}^{t,z}-\bold{P}^{t,z-1}\|\le \epsilon$
stop.\\ otherwise,\\
 set $z=z+1$ and go back to III.
\end{algorithm}

 To solve problem \eqref{p.o.c}, we use the dual method. Therefore, the corresponding Lagrangian function is given by \eqref{dual},
 \begin{figure*}[t]
\begin{align}\nonumber
& L(\mathbf{\tilde{p}},\boldsymbol{\delta},\boldsymbol{\beta})=\sum_{f\in \mathcal{F}}\sum_{m\in \mathcal{M}_f} \sum_{c\in \mathcal{C}}\xi^f_{m,c}\log(\gamma^f_{m,c}(\exp(\tilde{p}^f_{m,c})))+\psi ^f_{m,c}+ \sum_{f\in \mathcal{F}}\, \delta_f(p^f_{\text{max}} -\sum_{m\in \mathcal{M}_f} \sum_{n\in \mathcal{N}}\exp(\tilde{p}^f_{m,c}))\\\label{dual}&-\sum_{f1\in \mathcal{F}}\sum_{s\in \mathcal{M}_{f1}}\sum_{c \in\mathcal{C}}\sum_{j\in \mathcal{M}_f,|\hat{h}^f_{j,c}|^2>|\hat{h}^f_{s,c}|^2} \beta_{f1scj}\Big(-(\sum_{n\in \mathcal{N}}\rho^{f1}_{n,c}\eta^{f1}_{n,c}|h^{f1}_{m,n}|^2)(\sum_{f'\in\mathcal{F}/\{f\}}\sum_{m'\in \mathcal{M}_{f'}}\sum_{n\in \mathcal{N}} q^{f'}_{n,c}\\\nonumber&(\big({\exp{\tilde{p}_{m',c}^{f',t-1}}}+ {\exp{\tilde{p}_{m',c}^{f',t-1}}}(\tilde{p}_{m',c}^{f'}-\tilde{p}_{m',c}^{f',t-1})\big))\rho^{f'}_{n,c}\eta^{f'}_{n,c}|h^{f'}_{j,n}|^2+\sum_{i\in \mathcal{M}_f,|\hat{h}^f_{i,c}|^2>|\hat{h}^f_{s,c}|^2}\sum_{n\in\mathcal{N}}q^{f1}_{i,c}\big({\exp{\tilde{p}_{i,c}^{f1,t-1}}}+ {\exp{\tilde{p}_{i,c}^{f1,t-1}}}\\&\nonumber(\tilde{p}_{i,c}^{f1}-\tilde{p}_{i,c}^{f1,t-1})\big)\rho^{f1}_{n,c}\eta^{f1}_{n,c}|h^{f1}_{j,n}|^2+|w^{f1}_{j,c}|^2)+(\sum_{n\in \mathcal{N}}\rho^{f1}_{n,c}\eta^{f1}_{n,c}|h^{f1}_{j,n}|^2) (\sum_{f'\in\mathcal{F}/\{f\}}\sum_{m'\in \mathcal{M}_{f'}}\sum_{n\in \mathcal{N}} q^{f'}_{n,c}\exp{\tilde{p}}^{f'}_{m',c}\rho^{f'}_{n,c}\eta^{f'}_{n,c}|h^{f'}_{m',n}|^2+\\&\nonumber\sum_{i\in \mathcal{M}_f,|\hat{h}^f_{i,c}|^2>|\hat{h}^f_{s,c}|^2}\sum_{n\in\mathcal{N}}q^{f1}_{i,c}\exp{\tilde{p}}^{f1}_{i,c}\rho^{f1}_{n,c}\eta^{f1}_{n,c}|h^{f1}_{m,n}|^2+|w^{f1}_{m,c}|^2)\Big),
\end{align}
\hrule 
\end{figure*}
where
$\boldsymbol{\delta}$ and $\boldsymbol{\beta}$ are the   Lagrange multipliers.
The dual objective function is given by
\begin{equation}\label{dfrf}
g(\boldsymbol{\delta})=\max_{\mathbf{\tilde{p}}}L(\mathbf{\tilde{p}},\boldsymbol{\delta},\boldsymbol{\beta}).
\end{equation}
To solve the dual  problem, we should find the stationary point of \eqref{dual}
 with respect to
$\mathbf{\tilde{p}}$ where
$\boldsymbol{\delta}$ and $\boldsymbol{\beta}$
are fixed. Therefore, we have:
\begin{align}\label{eq000}
& \dfrac{\partial(L(\mathbf{\tilde{p}},\boldsymbol{\delta},\boldsymbol{\beta})}{\partial\tilde{p}^f_{m,c}}=0.
\end{align}
By simplifying \eqref{eq000}, $p^f_{m,c}$ is given by:
\begin{align}\label{mn}
p^f_{m,c}=\Bigg[\dfrac{ \xi^f_{m,c}+G^f_{m,c}}{\delta_f+A^f_{m,c}+B^f_{m,c}+C^f_{m,c}}\Bigg]^{+},
\end{align}
 where $[.]^+=\max(.,0)$ and
 \begin{align}\nonumber
 A^f_{m,c}=\sum^{M_f}_{i=m+1}\xi^f_{i,c}\dfrac{\gamma^f_{i,c}(p^f_{i,c})}{p^f_{i,c}},
 \end{align}
 \begin{align}\nonumber
   B^f_{m,c}=\sum_{k\in \mathcal{F}/\{f\}}\sum_{{i\in \mathcal{M}_f}}\xi^k_{i,j}\dfrac{\sum_{{n\in \mathcal{N}}}\rho^f_{n,c}\eta^f_{n,c}|h^f_{m,n}|^2\gamma^k_{i,j}}{\sum_{{n\in \mathcal{N}}}\rho^k_{n,c}\eta^k_{i,n}|h^k_{i,n}|^2p^k_{i,j}},
   \end{align}
   \begin{align}\nonumber
   &C^f_{m,c}=\sum_{f1\in \mathcal{F}}\sum_{s\in \mathcal{M}_{f1}}\sum_{j\in \mathcal{M}_f,|\hat{h}^f_{j,c}|^2 >|\hat{h}^f_{s,c}|^2}\beta_{f1scj}(\sum_{n\in \mathcal{N}}\\&\nonumber\rho^{f1}_{n,c}\eta^{f1}_{n,c}|h^{f1}_{j,n}|^2)(
 \sum_{{n\in \mathcal{N}}}q^{f}_{n,c}\rho^f_{n,c}\eta^{f}_{n,c}|h^{f}_{m,n}|^2
 \sum_{f1\in \mathcal{F}}\sum_{s\in \mathcal{M}_{f1}}\\&\nonumber\sum_{j\in \mathcal{M}_{f1},|\hat{h}^{f1}_{j,c}|^2<|\hat{h}^{f1}_{s,c}|^2} \beta_{f1scj}(\sum_{n\in \mathcal{N}}\rho^{f1}_{n,c}\eta^{f1}_{n,c}|h^{f1}_{j,n}|^2)\\&\nonumber
 \sum_{{n\in \mathcal{N}}} q^f_{m,c}\rho^f_{n,c}\eta^f_{n,c}|h^f_{s,n}|^2
 ),
 \end{align}
  and
  \begin{align}\nonumber
  &G^f_{m,c}=\sum_{f1\in \mathcal{F}/{f}}\sum_{s\in \mathcal{M}_{f1}}\sum_{j\in \mathcal{M}_f,|\hat{h}^f_{j,c}|^2>|\hat{h}^f_{s,c}|^2} \beta_{f1scj}\big(
 (\sum_{n\in \mathcal{N}}\rho^{f1}_{n,c}\\&\nonumber\eta^{f1}_{n,c}|h^{f1}_{s,n}|^2)(\sum_{n\in \mathcal{N}} q^{f}_{n,c}\big(     {\exp{\tilde{p}_{m,c}^{f,t-1}}}\big)\rho^{f}_{n,c}\eta^{f}_{n,c}|h^{f}_{j,n}|^2+\\&\nonumber\sum_{f1\in \mathcal{F}/{f}}\sum_{s\in \mathcal{M}_{f1}}\sum_{j\in \mathcal{M}_{f1},|\hat{h}^{f1}_{j,c}|^2<|\hat{h}^{f1}_{s,c}|^2} \beta_{f1scj}(\sum_{n\in \mathcal{N}}\rho^{f1}_{n,c}\eta^{f1}_{n,c}\\&\nonumber|h^{f1}_{s,n}|^2)(\sum_{n\in\mathcal{N}}q^f_{m,c}\big(     {\exp{\tilde{p}_{m,c}^{f,t-1}}}\big)\rho^{f}_{n,c}\eta^{f}_{n,c}|h^{f}_{m,n}|^2))
 \big).
 \end{align}

To update the dual multipliers, the subgradiant method is applied as follows:
\begin{equation}\label{update1}
\delta_f^{u+1}=[\delta_f^{u}-\nu_1(p_{\text{max}}^f-\sum_{m\in \mathcal{M}_f} \sum_{n\in \mathcal{N}}q^f_{m,n}p^f_{m,n})]^+,
\end{equation}
and
\begin{equation}\label{update2}
\beta_{f1scj}^{u+1}=[\beta_{f1scj}^{u}-\nu_2(\eqref{DCrem1})]^+,
\end{equation}
where  $u$ indicates  the iteration number of updating dual multipliers, and $\nu_1$ and $\nu_2$  are  small updating step-size.

The final algorithm in the power allocation sub-problem  is shown in  Algorithm \ref{table-3}.
\begin{algorithm}
\caption{ALGORITHM TO FIND STATIONARY POINT }
\label{table-3}
I:  Set $u=0$ and initialize  $\boldsymbol{\delta}^0$ and $\boldsymbol{\beta}^0$,\\
II: Repeat:
\\
III: Compute $\mathbf{p}$ by applying \eqref{mn}, \\
IV: Update $\boldsymbol{\delta}$ and $\boldsymbol{\beta}$,
  by using \eqref{update1} and \eqref{update2}, respectively, \\
V: When $\|\bold{P}^{u}-\bold{P}^{u-1}\|\le \epsilon$
stop.\\ otherwise,\\
 set $u=u+1$ and go back to III.
\end{algorithm}

\subsection{Codebook Assignment}
The problem of codebook assignment is formulated as:

\begin{align}\label{c_orj_p_sc_p}
&\max_{\mathbf{Q}}\; \sum_{f\in\mathcal{F}}\sum_{m\in \mathcal{M}_f}\sum_{c\in \mathcal{C}}r^f_{m,c}(\mathbf{Q}),\\& \nonumber
 \text{s.t.}:\hspace{.25cm}
\eqref{eeq8b}-\eqref{eeq8h}, \eqref{eeq8e}.\\&\nonumber
\end{align}

Problem \eqref{c_orj_p_sc_p} is an  INLP,  which can be solved using MADS algorithm. To apply MADS algorithm,   available optimization software  such as NOMAD solver \cite{bib36}  can be used.

\section{simulation results}\label{simulation resuls}
In this section, the system sum rate for PSMA, SCMA, and PD-NOMA is evaluated under different numbers of users and  small cells. In the
numerical results,  the system parameters are set as follows: MBS radius is $1$ Km, the SBSs radius are  $20$ m,
$N = 8$, $\eta^f_{n,c}= 1/2\,\, \forall f,c,m$, $S=2$, $K=6$, $h^f_{m,n} = x^f_{m,n}(d^f_{m})^{\mu}$ where $\mu$ indicates the path
loss exponent and $\mu = -2$, $x^f_{m,n}$ indicates the Rayleigh fading,
and $d^f_{m}$ demonstrates the distance between user $m$ and BS
$f$. 

Fig. \ref{pic-2} shows  the system sum rate versus the total number of users for PSMA, PD-NOMA, and SCMA, where we assume that
$P^1_{\text{max}} =30$ Watts,
$P^f_{\text{max}} = 2$ Watts for $f\in\{2,\dots,F\}$, and for PD-NOMA and PSMA $L_T=3$ users. As can be seen, PSMA  significantly improves the system performance compared to the other NOMA approaches. 
For smaller number of users, PSMA provides about 30\% more throughput. This increases to 50\% for larger number of users.

Fig. \ref{pic-4} shows the system sum rate versus the total transmit power for PSMA, PD-NOMA, and SCMA, where we considered that the total number of users is $12$, and for PD-NOMA and PSMA $L_T=3$ users. Again PSMA outperforms SCMA and PD-NOMA by a great margin.

Finally in Fig. \ref{pic-3}, we investigated the effect of $L_T$, the total number of users that can be assigned to a codebook simultaneously, on PSMA performance. As can be seen, PSMA with $L_T=1$, exhibits a performance simalr to SCMA. As $L_T$ increases, the performance improves, roughly $25\%$ improvement per each unit increase in $L_T$. 
We, however, expect that this saturates at some point. Moreover, larger $L_T$ implies more comalcity $5\%$.

\begin{figure}
\centering
\includegraphics[width=.5\textwidth]{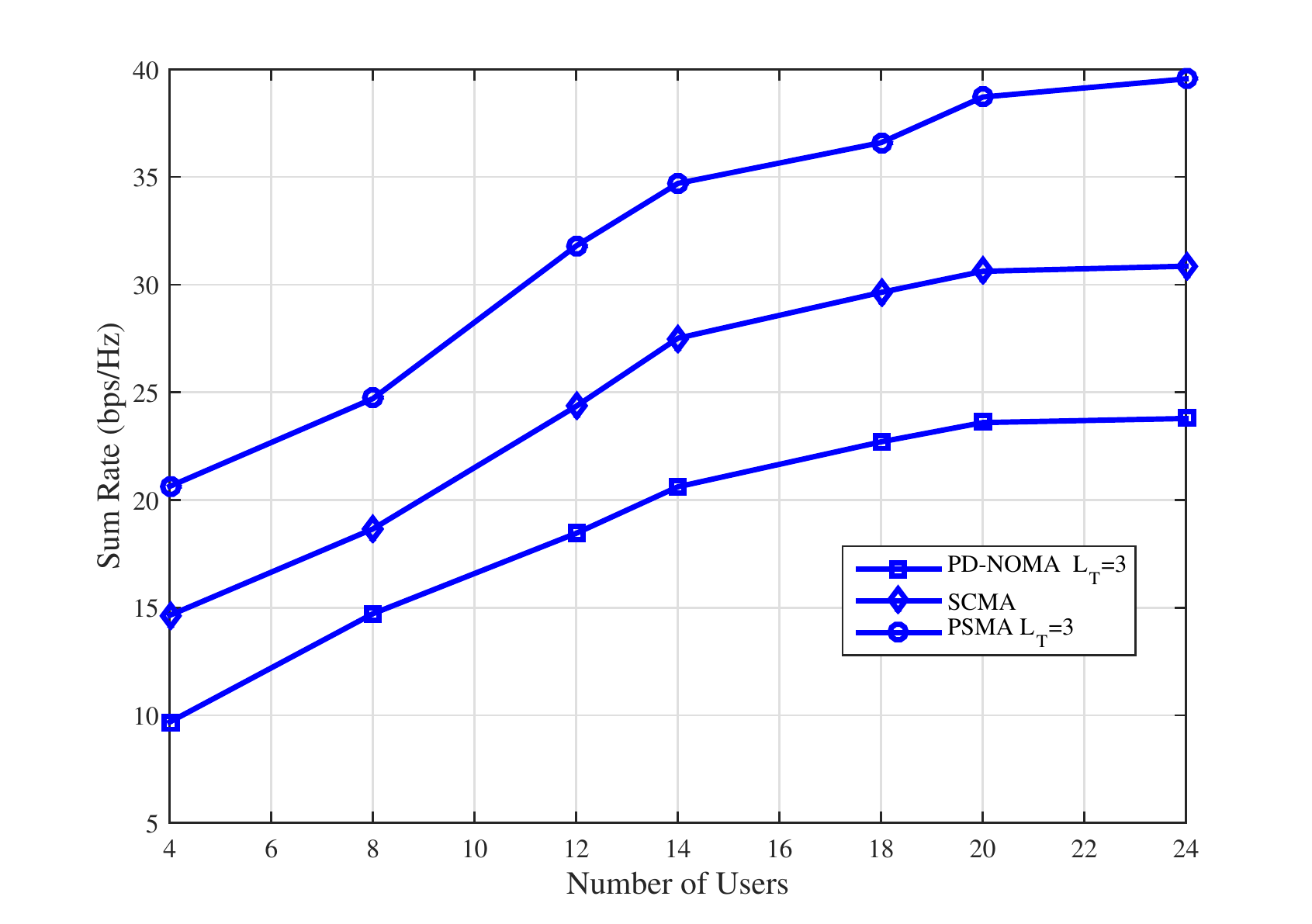}
\caption{System sum rate versus number of users. }
\label{pic-2}
\end{figure}
\begin{figure}
\centering
\includegraphics[width=.5\textwidth]{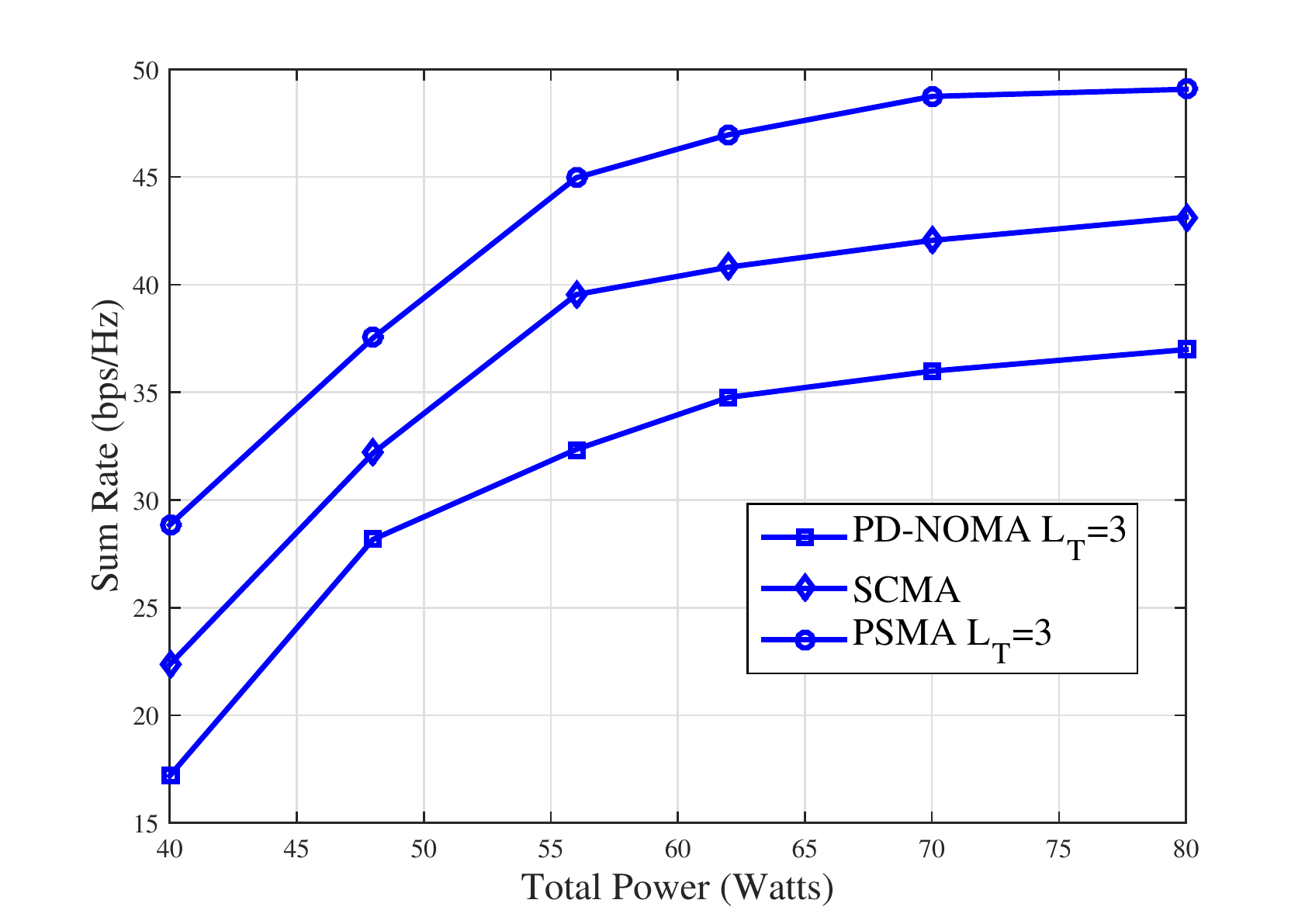}
\caption{System sum rate versus total transmit power. }
\label{pic-3}
\end{figure}

\begin{figure}
\centering
\includegraphics[width=.5\textwidth]{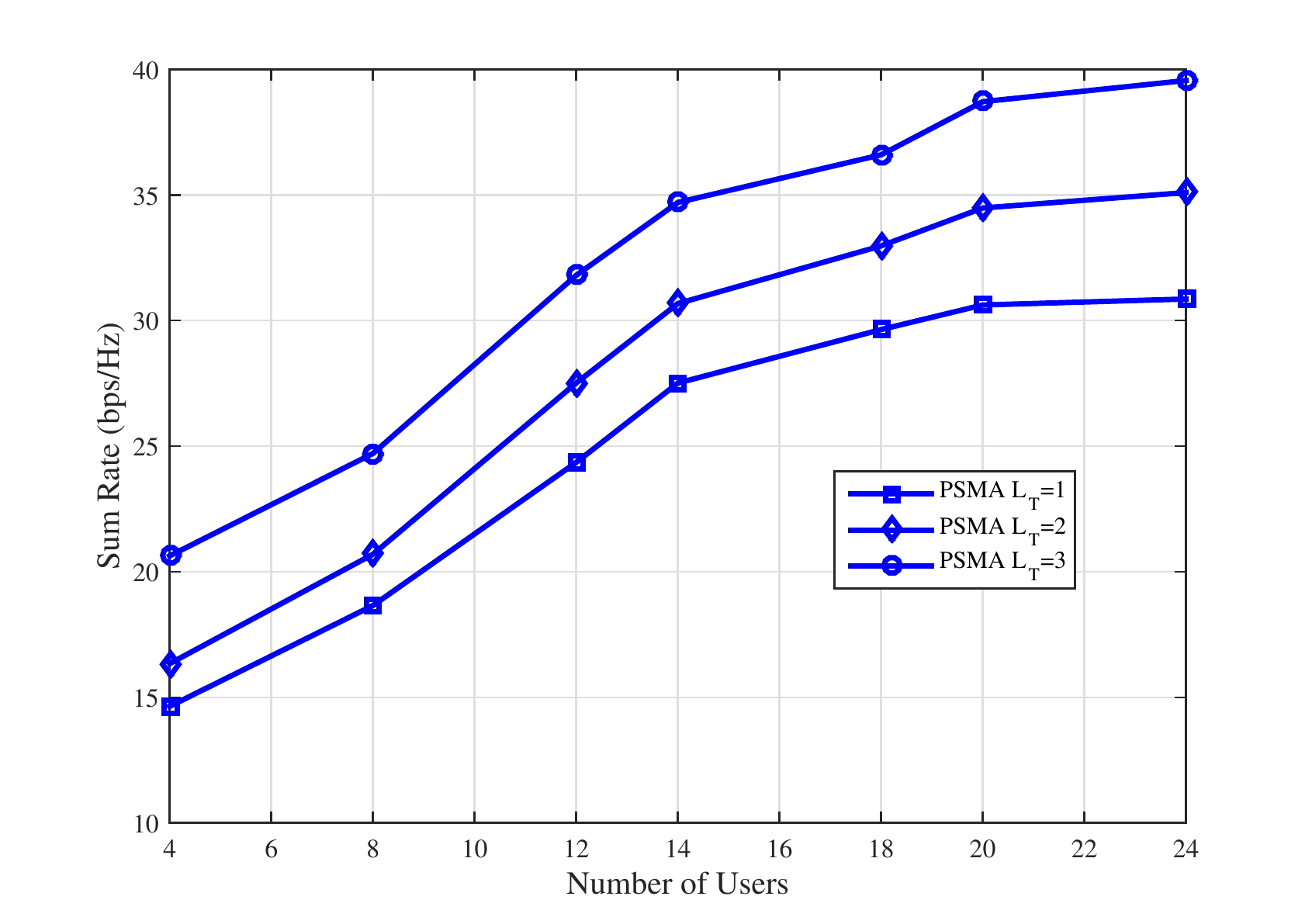}
\caption{ System sum rate versus $L_T$ (total number of users that can be assign to a codebook. )}
\label{pic-4}
\end{figure}

\section{conclusion}\label{CONCLUSION}
In this paper, we proposed a new MA technique for 5G which uses code and power domain to send multiple users' signals in a subcarrier. We investigated the PSMA transmitter and receiver  and compared it to other NOMA approaches from the aspect of receiver complexity and system performance. To this end, we proposed a novel resource allocation problem. To solve the proposed problem, we used an iterative algorithm based on the SCA approach where in each iteration, codebook assignment was solved by applying the MADS algorithm and power allocation was solved based on the SCALE and DC methods. Moreover, from simulation results, we concluded that the PSMA technique significantly outperforms other NOMA techniques while imposing a reasonable increase in complexity to the system. Future works include improving the robustness of PSMA decoders by taking into account factors such as CSI uncertainty, and developing more robust SIC ordering techniques. As a future work, we study  the link level performance of PSMA based systems and compare it to
PD-NOMA and SCMA based systems.

\end{document}